\documentclass[a4paper,UKenglish,cleveref, autoref]{lipics-v2019}
 
\usepackage{microtype}
\usepackage{multirow}
\usepackage{float}
\usepackage{helvet}
\usepackage{color}
\usepackage{algorithm}
\usepackage{fdsymbol}
\usepackage{booktabs}
\usepackage[noend]{algpseudocode}
\usepackage{caption}

\graphicspath{{./figures/}}

\newcommand{\range}[1]{\overset{\leftrightblackspoon}{#1}}

\title{Simulating the DNA Overlap Graph in Succinct Space}
\titlerunning{Simulating the DNA overlap graph in succinct space} 

\author{Diego Díaz-Domínguez\footnote{Corresponding author}}{CeBiB — Center for Biotechnology and Bioengineering, \and Department of Computer Science, University of Chile, Chile}{diediaz@dcc.uchile.cl}{https://orcid.org/0000-0002-9071-0254}{}

\author{Travis Gagie}{School of Computer Science and Telecommunications, Diego Portales
University, Chile \and CeBiB — Center for Biotechnology and Bioengineering}{travis.gagie@gmail.com}{https://orcid.org/0000-0003-3689-327X}{}

\author{Gonzalo Navarro}{CeBiB — Center for Biotechnology and Bioengineering, \and Department of Computer Science, University of Chile, Chile}{gnavarro@dcc.uchile.cl}{https://orcid.org/0000-0002-2286-741X}{}

\authorrunning{D. Díaz-Domínguez et al.} 

\Copyright{Diego Díaz-Domínguez, Travis Gagie and Gonzalo Navarro}

\ccsdesc[100]{Applied computing~Computational biology}
\ccsdesc[100]{Information systems~Data compression}

\keywords{Overlap graph, de Bruijn graph, DNA sequencing, Succinct ordinal trees}

\category{}

\relatedversion{}

\supplement{}

\funding{Partially supported by Basal Funds FB0001, Conicyt, Chile; by a Conicyt Ph.D. Scholarship; by Fondecyt Grants 1-171058 and 1-170048, Chile; and by the European Union’s Horizon 2020 research and innovation programme under the Marie Sklodowska-Curie [grant agreement No 690941]}

\acknowledgements{We thank the reviewers for their helpful comments}

\nolinenumbers 

\EventEditors{Nadia Pisanti and Solon P. Pissis}
\EventNoEds{2}
\EventLongTitle{30th Annual Symposium on Combinatorial Pattern Matching (CPM 2019)}
\EventShortTitle{CPM 2019}
\EventAcronym{CPM}
\EventYear{2019}
\EventDate{June 18--20, 2019}
\EventLocation{Pisa, Italy}
\EventLogo{}
\SeriesVolume{128}
\ArticleNo{27}

\begin{document}

\maketitle

\begin{abstract}
Converting a set of sequencing reads into a lossless compact data structure that encodes all the relevant biological information is a major challenge. The classical approaches are to build the string graph or the \emph{de Bruijn} graph (dBG) of some order $k$. Each has advantages over the other depending on the application. Still, the ideal setting would be to have an index of the reads that is easy to build and can be adapted to any type of biological analysis. In this paper we propose \textit{rBOSS}, a new data structure based on the Burrows-Wheeler Transform (BWT), which gets close to that ideal. Our \textit{rBOSS} simultaneously encodes all the dBGs of a set of sequencing reads up to some order $k$, and for any dBG node $v$, it can compute in $\mathcal{O}(k)$ time all the other nodes whose labels have an overlap of at least $m$ characters with the label of $v$, with $m$ being a parameter. If we choose the parameter $k$ equal to the size of the reads (assuming that all have equal length), then we can simulate the overlap graph of the read set. Instead of storing the edges of this graph explicitly, \emph{rBOSS} computes them on the fly as we traverse the graph. As most BWT-based structures, \textit{rBOSS} is unidirectional, meaning that we can retrieve only the suffix overlaps of the nodes. However, we exploit the property of the DNA reverse complements to simulate bi-directionality. We implemented a genome assembler on top of \textit{rBOSS} to demonstrate its usefulness. The experimental results show that, using $k=100$, our \textit{rBOSS}-based assembler can process $\sim$500K reads of 150 characters long each (a \texttt{FASTQ} file of 185 MB) in less than 15 minutes and using 110 MB in total. It produces contigs of mean sizes over 10,000, which is twice the size obtained by using a pure de Bruijn graph of fixed length $k$.
\end{abstract}

\section{Introduction}\label{sec:introduction}

Obtaining and extracting the relevant information from a collection of DNA sequencing \emph{reads}\footnote{A string that represents the inferred sequence of base pairs in a segment of a DNA molecule.}, for assembly and other analysis purposes, usually requires a lot of time and space. The techniques for compressed indexing developed in recent years (see~\cite{navarro2007compressed} for a full review) have significantly contributed to reduce the computational costs. There is still no technique, however, that can preprocess the reads and represent all the relevant information in succinct space so that it can be used effectively.

The classical plain, and lossless, data structure to analyze reads is the \emph{overlap graph}. In this model, each node represents a particular read, and two nodes $v$ and $v'$ are connected by an edge with weight $o \ge m$ if $o$ is the maximum length of a suffix of $v$ that matches a prefix of $v'$, where $m$ is a parameter to filter out spurious overlaps. Computing the overlap graph from a set of reads is not difficult: it can be built from the suffix tree of the set or even from its \emph{Burrows-Wheeler Transform (BWT)}~\cite{BW94,simpson2010efficient} or the Longest Common Prefix array (LCP)~\cite{bonizzoni2014constructing}. Representing the graph, however, is expensive: a quadratic number of edges may have to be stored. A popular solution is to perform \emph{structural compression} over the graph by removing the transitive edges. The resulting graph is usually called the \emph{string graph} \cite{myers2005fragment,simpson2012efficient} or the \emph{irreductible overlap graph} \cite{makinen2015genome}. This approach, however, limits the applications of the data strucutre, because it removes information from the graph that can be useful.

Historically, string graphs have been used mainly in the context of \emph{genome assembly} \cite{denisov2008consensus,myers2005fragment,zimin2013masurca}, but as sequencing datasets have grown over the years, they have been discarded in favor of other lossy, but more succinct, representations. The most famous of these representations is the \emph{de Bruijn} graph (dBG). This data structure encodes the relationship between all the substrings of length $k$ in the set. A dBG is easy to construct and it can be represented succinctly~\cite{BOSS12}. Besides,  it encodes the context of the substrings of length $k-1$ in its topology. Thus, for instance, if a substring of length $k-1$ appears in several contexts of the set, then its dBG node will have several edges. As for the string graph, the first application of the dBG was the assembly of genomes~\cite{bankevich2012spades,li2015megahit,Peng2010,zerbino2008velvet}, but through the years its use has been extended to other kind of analyses~\cite{bray2016near,iqbal2012novo}.

The disadvantage of dBGs, however, is that they are lossy, because the information we can retrieve from them is limited by $k$. A way to overcome some of the restrictions imposed by $k$ is to add variable-order functionality to the dBG, that is, to encode several dBGs with different values for $k$ at the same time. The contexts of the graph can then be shortened or lengthened on demand, depending on the need to have more or less edges from a node. Some succinct dBG representations supporting variable-order functionality up to some maximum order $k$ have been proposed~\cite{BBGPS15}, but they increase the space requirements by a $\log k$ factor. Besides, even using variable-order functionality, the data structure remains lossy, because the order of the graph can be lengthened only up to $k$. By choosing $k$ equal to the size of the reads, the variable-order dBG becomes lossless, and equivalent to the overlap graph, but the $\log k$ factor becomes significant for the typical read sizes.  

Almost every analysis over DNA sequencing data can be reduced to looking for suffix-prefix overlaps between the reads, and in this regard the dBG has adapted well to many bioinformatic applications because it is a lightweight (lossy) representation of the overlaps. Still, 
searching for biological signals in a dBG requires the detection of complex graph substructures such as bubbles, super bubbles, tips, and so on, and those can be expensive to find. The overlap graph, on the other hand, is a much simpler model. It can be adapted to different applications other than assembly in the same way as the dBG, but it has the advantage of being lossless and not requiring too much preprocessing after its construction. The problem, as stated before, is the space to encode the edges of the graph. A better approach would be to have a compact data structure that can quickly compute the overlaps (i.e., the edges) on the fly instead of storing them explicitly. Such a solution would require a moderate preprocessing of the read set, would retain all the information, and would use a reasonable  amount of space.

\subparagraph{Our contribution. \rm We address the problem of succinctly representing and analyzing a collection $\mathcal{R}$ of sequencing reads. To this end, we define a new compact data structure we call {\em rBOSS}. It is an intermediate structure between a dBG and an overlap graph, because it can compute the context of the sequences in the same way a dBG does, but it can also compute on the fly the overlaps between different substrings of $\mathcal{R}$. The \textit{rBOSS} index is based on \textit{BOSS} \cite{BOSS12}, a BWT-based representation of dBGs, which is augmented with a tree we call the {\em overlap tree}. This tree increases the size of the data structure by $4n + o(n)$ bits, where $n$ is the number of nodes in the dBG encoded by \textit{BOSS}. By choosing $k$ equal to the length of the reads (which we assume to have all the same length), we can simulate in compressed space an overlap graph whose edges have a weight $o \ge m$, where $m$ is given as a parameter. The simulation of the graph builds on the basic primitives {\tt nextcontained} and {\tt buildL}. Our overlap tree reduces their time complexity from $\mathcal{O}(k^2)$ to $\mathcal{O}(1)$ and $\mathcal{O}(k)$, respectively.\\}

In addition to \textit{rBOSS}, we also formalize the idea of weighting the overlap graph edges according to transitive connections, and explain how this new weighting scheme can be used to solve biological problems other than assembly. To our knowledge, this is the first time a measure of this kind is proposed for overlap graphs. Finally, we demonstrate the usefulness of \textit{rBOSS} by implementing a genome assembler on top of it. Our experimental results show that, by using $k=100$, the assembler can process $\sim$500K reads of 150 characters long each in less than 15 minutes and using 110 MB in total. It produces contigs of mean sizes over 10,000, which is twice the size obtained by using a pure dBG of fixed length $k$.

\section{Preliminaries}\label{prems}

\subparagraph{DNA strings. \rm A DNA sequence $R$ is a string over the alphabet $\Sigma=\{\texttt{a},\texttt{c},\texttt{g},\texttt{t}\}$ (which we map to $[2..\sigma]$), where every symbol represents a particular nucleotide in a DNA molecule. The DNA \emph{complement} is a permutation $\pi[2..\sigma]$ that reorders the symbols in $\Sigma$ exchanging \texttt{a} with \texttt{t} and \texttt{c} with \texttt{g}. The \emph{reverse complement} of $R$, denoted $R^{rc}$, is a string transformation that reverses $R$ and then replaces every symbol $R[i]$ by its complement $\pi(R[i])$. For technical convenience we add to $\Sigma$ the so-called \emph{dummy} symbol \$, which is always mapped to 1.} 

\subparagraph{De Bruijn graphs. \rm A de Bruijn Graph (dBG)~\cite{de1946combinatorial} of order $k$ of a set of strings $\mathcal{R} = \{ R_1, R_2, \ldots, R_r\}$, $DBG_{\mathcal{R},k}$, is a labeled directed graph $G=(V, E)$ where every node $v \in V$ is labeled by a distinct substring of $\mathcal{R}$ of length $k-1$, and every edge $(v, u) \in E$ represents the substring $S[1..k]$ of $\mathcal{R}$ such that $v$ is labeled by the prefix $S[1..k-1]$, $u$ is labeled by the suffix $S[2..k]$, and the label of $(v,u)$ is the symbol $S[k]$. We identify a node with its label.\\}

A \emph{variable-order} dBG (vo-dBG)~\cite{BBGPS15}, $voDBG_{\mathcal{R},k}$, is formed by the union of all the graphs $DBG_{\mathcal{R},k'}$, with $1\leq k' \leq k$. Each $DBG_{\mathcal{R},k}$ represents a \emph{context} of $voDBG_{\mathcal{R},k}$. In addition to the (directed) edges of each $DBG_{\mathcal{R},k}$, two nodes $v \in DBG_{\mathcal{R},k'}$ and $v' \in DBG_{\mathcal{R},k''}$, with $k' > k''$, are connected by an undirected edge $(v, v')$ if $v'$ is a suffix of $v$. Following the edge $(v, v')$ or $(v',v)$ is called a \emph{change of order}. We then identify node order with length.

\subparagraph{BOSS representation for de Bruijn graphs. \rm \textit{BOSS}~\cite{BOSS12} is a succinct data structure, similar to the \textit{FM-index}~\cite{ferragina2005indexing}, for encoding dBGs. In \textit{BOSS}, the nodes are represented as rows in a matrix of $k-1$ columns, and are sorted in reverse lexicographical order (i.e., reading the labels right to left). All the edge (one-symbol) labels of the graph are stored in a unique sequence $E$ sorted by the \textit{BOSS} order of the source nodes, so the symbols of the outgoing edges of each node fall in a contiguous range. A bitmap $B$ of size $e=|E|$ marks the last outgoing symbol in $E$ of every dBG node. Finally, an array $C[1..\sigma]$ stores in $C[i]$ the number of node labels that end with a symbol lexicographically smaller than $i$.\\}\label{boss}

Prefixes in $\mathcal{R}$ of size $d<k$ are artificially represented in  \textit{BOSS} as strings of length $k$ padded at the left with $k-d$ symbols \$. Equivalently, suffixes of size $d<k$ are represented as strings of length $k$ padded at the right with $k-d$ symbols \$. For this work, however, suffixes of size $d<k$ are not necessary. Strings formed only by symbols \$ are also called dummy.

The complete index is thus composed of the vectors $E$, $C$, and $B$. It can be stored in $e(\mathcal{H}_{0}(E)+\mathcal{H}_{0}(B))(1+o(1)) + \mathcal{O}(\sigma \log n)$ bits, where $\mathcal{H}_0$ is the zero-order empirical entropy~\cite[Sec~2.3]{navarro2016compact}. This space is reached with a Huffman-shaped Wavelet Tree \cite{MN05} for $E$, a compressed bitmap~\cite{raman2007} for $B$ (as it is usually very dense), and a plain array for $C$. 

 \textit{BOSS} supports several navigational queries, most of them within $O(\log \sigma)$ time or less. The most relevant ones for us are:

\begin{itemize}
    \item \texttt{outdegree$(v)$}: number of outgoing edges of $v$.
    \item \texttt{forward$(v, a)$}: node reached by following an edge from $v$ labeled with symbol $a$.
    \item \texttt{indegree$(v)$}: number of incoming edges of $v$.
    \item \texttt{backward$(v)$}: list of the nodes with an outgoing edge to $v$.
\end{itemize}

Boucher et al.~\cite{BBGPS15} noticed that by considering just the last $k'-1$ columns in the  \textit{BOSS} matrix, with $k'\leq k$, the resulting nodes are the same as those in the dBG of order $k'$. To allow changing the order of the dBG in  \textit{BOSS}, they augmented the data structure with the \emph{longest common suffix} ($LCS$) array. The $LCS$ array stores, for every node of order $k$, the size of the longest suffix shared with its predecessor node in the  \textit{BOSS} matrix. They called this new index the \emph{variable-order}  \textit{BOSS} (\textit{VO-BOSS}), which supports the following additional operations:

\begin{itemize}
    \item \texttt{shorter$([i,j], k')$}: range of the nodes suffixed by the last $k'$ characters of $v$. 
    \item \texttt{longer$([i,j], k')$}: list of the nodes of length $k'$ that end with $v$. 
    \item \texttt{maxlen$([i,j], a)$}: a node in the index suffixed by $v$, and that has an outgoing edge labeled $a$.
\end{itemize}

Where $[i,j]$ is the range of rows in the \textit{BOSS} matrix suffixed by the label of a vo-dBG node $v$. By using a Wavelet Tree~\cite{grossi2003}, the $LCS$ can be stored in $n\log k + o(n\log k)$ bits, the function \texttt{shorter($[i,j],k'$)} can be answered in $\mathcal{O}(\log k)$ time and the function \texttt{longer($[i,j],k'$)} in $\mathcal{O}(|U|\log k)$ time, where $U$ is the set of rows of the \textit{BOSS} matrix contained within the range $[i-1..j]$, and whose $LCS$ values are below $k'$. The function \texttt{maxlen($[i,j],a$)} is implemented using the arrays $B$ and $E$ from $BOSS$, and hence it is answered in $\mathcal{O}(\log \sigma)$ time.

\subparagraph{Succinct representation of ordinal trees. \rm
An ordinal tree $T$ with $n$ nodes can be stored succinctly as a sequence of \emph{balanced parentheses} ($BP$) encoded as a bit vector $B[1..2n]$. Every node $v$ in $T$ is represented by a pair of parentheses \texttt{(..)} that contain the encoding of the subtree rooted at $v$. Every node of $T$ can be identified by the position in $B$ of its open parenthesis. Many navigational operations over $T$ can be simulated over $B$ in constant time, by using a structure that requires $o(n)$ bits on top of $B$~\cite{NS14}. }\label{ordtrees} 

\section{rBOSS}\label{rboss_sec}

\subparagraph{Basic definitions. \rm Let $\mathcal{R}=\{R_1...R_r\}$ be a collection of $r$ reads (strings) of length $z$ and let $\mathcal{R}^{rc}$ be a collection, also of $r$ strings, every $R_{i}' \in \mathcal{R}^{rc}$ being the reverse complement of one $R_j \in \mathcal{R}$. Aditionally, we define the set $\mathcal{R}^* = \mathcal{R} \cup \mathcal{R}^{rc}$. Let us denote $G=DBG_{\mathcal{R}^*,k}$ and $G'=voDBG_{\mathcal{R}^{*},k}$. A \emph{traversal} $P$ over $G$, or $G'$, is a sequence $(v_0,e_0,v_1...v_{t-1},e_t,v_t)$ where $v_0,v_1,...v_{t-1},v_t$ are nodes and $e_1..e_t$ are edges, $e_i$ connecting $v_{i-1}$ with $v_i$. $P$ will be a \emph{path} if all the nodes are different, except possibly the first and the last. In such case, $P$ is said to be a \emph{cycle}. $P$ is \emph{unary} if the nodes ($v_0..v_{t-1}$) have outdegree 1 and the nodes ($v_1..v_t$) have indegree 1. $P$ will be a \emph{right} traversal over $G$ or $G'$ if all its edges $e_i$ are directed from $v_{i-1}$ to $v_i$, and a \emph{left} traversal if all its edges $e_i$ are directed from $v_i$ to $v_{i-1}$. The string formed by the concatenation of the edge symbols of $P$ is referred to as its label. $P$ will be \emph{safe} if it is a path or a cycle and its label appears in $\mathcal{R}^*$ as a substring of some read or if it can be generated by overlapping two or more ${R_i..R_j} \in \mathcal{R}^{*}$ in tandem and then taking the string that results from the union of those reads. The overlaps between the reads have to be of minimum size $m$.\\}\label{defs}

Let $BOSS(G)$ and $VO$-$BOSS(G')$ be the \textit{BOSS} and \textit{VO-BOSS} indexes, respectively, for the graphs $G$ and $G'$. In both cases, the matrix with the $(k-1)$-length node labels is referred to as $M_{\mathcal{R}^*,k}$, or just $M$ when the context is clear. In $VO$-$BOSS(G')$, the range of rows in $M$ suffixed by $v$ is denoted $\range{v}$.

Rows of $M$ representing substrings of size $k-1$ in $\mathcal{R}^*$ are called \emph{solid} nodes and rows representing artificial $(k-1)$-length strings padded with dummy symbols from the left, and that represent prefixes in $\mathcal{R}^*$, are called \emph{linker} nodes. For a linker node $v$, the function  \texttt{llabel($v$)} returns the non-\$ suffix of $v$. A solid node that appears as a suffix in $\mathcal{R}$ is called an \emph{s-node} and a solid node that appears as a prefix in $\mathcal{R}^*$ is called a \emph{p-node}. A linker node $v$ is said to be \emph{contained} within another node $v'$ (solid or linker) if \texttt{llabel($v$)} is a suffix of $v'$.

An overlap of size $o$ between two solid nodes $v$ and $v'$, denoted $v\oplus^{o}v'$, occurs when the $o$-length suffix of $v$ is equal to the $o$-length prefix of $v'$. Relative to $v$, $v\oplus^{o}v'$ is a \emph{forward} overlap and $v'\oplus^{o}v$ is a \emph{backward} overlap. An overlap $v\oplus^{o}v'$ is \emph{valid} if (i) $m\leq o<k-2$, $m$ being some parameter, and $v'$ is a p-node, or (ii) $o=k-2$ and $v'$ is a solid node of any kind. Notice that case (ii) is equivalent to the definition of two dBG nodes connected by an edge. The overlap $v\oplus^{o}v'$ is considered \emph{transitive} if there is another solid node $v''$ with valid overlaps $v\oplus^{o'}v''$ and $v'\oplus^{o''}v''$, with $o' > o''$. If there is only one $v''$, then $v\oplus^{o'}v'$ is transitive and \emph{unique}. If such $v''$ does not exist, then $v\oplus^{o}v'$ is \emph{irreductible}. The string formed by the union of the solid nodes $v$ and $v'$ is denoted \texttt{label}$(v\oplus^{o}v')$. 

\subparagraph{Link between variable-order and overlaps. \rm Overlaps between reads in $\mathcal{R}^*$ can be computed using $VO$-$BOSS(G')$ as follows: extend a unary path using solid nodes as much as possible, and if a solid node $v$ without outgoing edges is reached, then decrease its order with \texttt{shorter} to retrieve the vo-dBG nodes that represent both a suffix of $v$ and a prefix of some read in $\mathcal{R}$. From these nodes, retrieve the overlapping solid nodes of $v$ by using \texttt{forward} and continue the graph right traversal from one of them.\\}

In  \textit{VO-BOSS}, however, \texttt{shorter} does not ensure that the label of the output node appears as a prefix in $\mathcal{R}^*$. The next lemma precises the condition that must hold to ensure this.

\begin{lemma}\label{l1}
In  \textit{VO-BOSS}, applying the operation {\tt shorter} to a node $v$ of order $k'\leq k$ will return a node $u$ of order $k''<k'$ that encodes a forward overlap for $v$ iff $\range{u}[1]$ is a linker node contained by $v$.
\end{lemma}

\begin{proof}

If all the left contexts in $\range{u}$ are non-dummy strings, then $u$ does not appear as a prefix in $\mathcal{R}^*$, and hence, following none of its edges will lead to a valid overlap of $v$. On the other hand, if a suffix $u$ of $v$ appears as a prefix in $\mathcal{R}^*$, then there is a node $v'$ in $G'$ at order $k$ whose label is formed by the concatenation of a dummy string and $u$, and that by definition is a linker node contained by $v$. Since elements in $\range{u}$ are sorted by the left contexts, $v'$ is placed in $\range{u}[1]$, because the dummy string is always the lexicographically smallest.
\end{proof}

Lemma~\ref{l1} allows us to find overlapping nodes that are not directly linked via edges in the dBG, by looking in smaller dBG orders. We formalize this operation as follows, where we look for the longest valid suffix, that is, the one with maximum lexicographic index. 

\begin{itemize}
\item $\texttt{nextcontained}(v)$: returns the greatest linker node $v'$, in lexicographical order, whose \texttt{llabel} represents both a suffix of $v$ and a prefix of some other node in $G'$\@.
\end{itemize}

\begin{theorem}
There is an algorithm that solves {\rm \tt nextcontained} in $\mathcal{O}(k^2 \log \sigma)$ time.
\end{theorem}

\begin{proof}
Incrementally decrease the order of $v$ by one until reaching a node $u$ with $\range{u}\supset \range{v}$, and that satisfies Lemma~\ref{l1}. If such $u$ exists, return it. If the order of $v$ decreases below $m$ before finding $u$, then $v$ does not contain any linker node $v'$ with $|\texttt{llabel}(v')| \ge m$. In such case, a dummy vo-dBG node is returned. The function reduces the order up to $k-2$ times. In each iteration, the operations \texttt{shorter} and \texttt{llabel} (to check Lemma \ref{l1}) are used, which take $\mathcal{O}(\log k)$ and $\mathcal{O}(k\log\sigma)$ time, respectively. Thus, the total time is $\mathcal{O}(k^{2}\log\sigma)$. 
\end{proof}

Notice, however, that a vo-dBG node in $G'$ might have more than one contained linker node, and those linkers whose \texttt{llabel} is of length $\ge m$ represent edges in the overlap graph. A useful operation is to build a set $L$ with all those relevant linkers. We can then follow the outgoing edges of every $l \in L$ to infer the solid nodes that overlap $v$ by at least $m$ symbols.

\begin{itemize}
\item \texttt{buildL($v$,$m$)}: the set of all the linker nodes contained by $v$ that represent a suffix of $v$ of length $\ge m$. 
\end{itemize}

Function \texttt{buildL} applies \texttt{nextcontained} iteratively until reaching a node $u$ that contains a linker node whose \texttt{llabel} has length below $m$. The rationale is that if $v$ contains $v'$, and in turn $v'$ contains $v''$, then $v$ also contains $v''$. 
Algorithm~\ref{a1} (in Appendix~\ref{pseudo}) shows the details.

Note that, if we chose $k=z+1$ to build \textit{VO-BOSS}$(G')$, then we are simulating the full overlap graph in compressed space. The edges are not stored explicitly, but computed on the fly by first obtaining $L=\texttt{buildL}(v,m)$, and then following the dBG outgoing edges of every $l \in L$. Still, the complexities of the involved operations \texttt{nextcontained} makes \textit{VO-BOSS} slow for exhaustive traversals, which is our main interest. We design a faster scheme in which follows; Figure~\ref{fig:rboss} exemplifies the various concepts. 

\subparagraph{A compact data structure to compute overlaps. \rm 
The function \texttt{buildL} can be regarded as a bottom-up traversal of the trie $T$ induced by the $(k-1)$-length labels of $M$ read in reverse. Every trie node $t$ corresponds to a vo-dBG node $v$ whose order is the string depth of $t$. The traversal starts in the trie leaf $t$ corresponding to the vo-dBG node $v$ given to \texttt{buildL}, and continues upward until finding the last ancestor $t'$ of $t$ with string depth $\ge m$. The movement from $t$ to $t'$ can be regarded as a sequence of applications of \texttt{nextcontained}. In each such application, we move from a node $t$ to its nearest ancestor $t'$ that is maximal (i.e., has more than one child) and whose leftmost child edge is labeled by a \$. \\}

Since non-maximal nodes in $T$ are not relevant for building $L$, the function \texttt{nextcontained} can be reimplemented using the topology of the \emph{compact} trie $T$ (i.e., collapsing unary paths) represented with $BP$ (see Section \ref{ordtrees}) instead of using \texttt{shorter}. In this way, we can get rid of the $LCS$ structure of  \textit{VO-BOSS}.

The resulting \textit{rBOSS} index can be built in linear time, as detailed in Appendix~\ref{building}.

Replacing the $LCS$ with the topology of $T$ in $BP$ poses two problems, though. First, it is not possible to define a minimum dBG order $m$ from which overlaps are not allowed, and second, Lemma \ref{l1} cannot be checked. Both problems arise because, unlike the $LCS$, the $BP$ data structure does not encode the string depths of the tree nodes (and thus, the represented node lengths). Still, the topology of $T$ can be reduced to precisely the nodes of interest for \texttt{nextcontained}, and thus avoid any check. We call this structurally-compressed version of $T$ the \emph{overlap tree}.

\begin{figure}[t]
\centering
\includegraphics[width=0.87\linewidth]{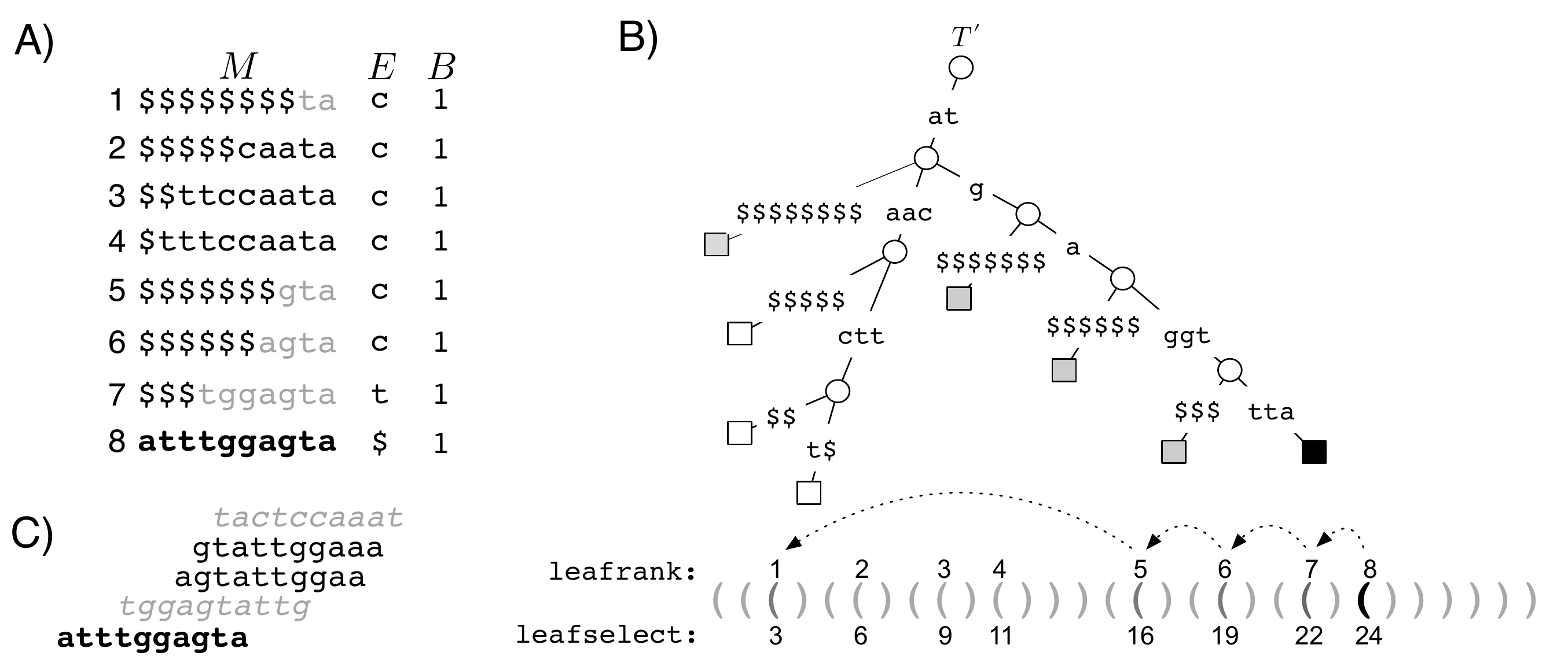}
\caption{A segment of \textit{rBOSS} for the set $\mathcal{R}$=\{\texttt{atttggagta},\texttt{gtattggaaa},\texttt{agtattggaa},\texttt{caatactcca}\}, with parameters $k=11$ and $m=2$. The computation of the forward overlaps for $R_1=\mathcal{R}[1]$ are also shown. A) range of $M$ that includes the solid node $v$ that represents sequence $R_1$ (sequence in bold in row 8) and its contained nodes (grayed rows). B) subtree of $T'$ that maps the range of $M$ in A). The black leaf maps $v$ in $T'$ and the gray leaves map the linker nodes contained by $v$. The array of parentheses below represents the same subtree in $BP$. The rank of every leaf (\texttt{leafrank}) is shown above its relative position in the $BP$ array (\texttt{leafselect}). Each dashed arrow represents a call of the function \texttt{nextcontained} and all the arrows together stand for the function \texttt{buildL}. C) Forward overlaps of $v$. Gray sequences in italic are those reads where the reverse complement matches with $v$.}
\label{fig:rboss}
\end{figure}

\begin{theorem}\label{t2}
There is a structure using $4 + o(1)$ bits per dBG node that implements the function {\rm\tt nextcontained} in $\mathcal{O}(1)$ time and the function {\rm\tt buildL} in time $\mathcal{O}(k)$.
\end{theorem}

\begin{proof}
The structure is the $BP$ encoding of a tree $T'$ that is obtained by removing some nodes from the compact trie $T$ ($T$ has one leaf per dBG node, and less than 2 nodes per dBG node because it is compact). First, all the internal nodes of $T$ with string depth below $m$ are discarded, and the subtrees left are connected to the root of $T'$. Second, every internal node $t' \in T$ whose  leftmost-child edge is not labeled by a dummy string is also discarded, and its children are recursively connected to the parent of $t'$. Note that all the leaves of $T$ are in $T'$.

Therefore, $T'$ has precisely the nodes of interest for operation {\rm\tt nextcontained}$(v)$. We simply find the $i$th left-to-right leaf $t$ in $T'$, where $i$ is the row of $M$ corresponding to $v$ (note that rows of $M$ and leaves of $T$ and $T'$ are in the same order). Then, we move to the parent $t'$ of $t$ and return its leftmost child. An exception occurs if the leftmost child of $t'$ is precisely $t$, which means that $v$ is a linker node and thus its next contained node is the leftmost child of the parent of $t'$. Finally, we return the rank of the desired leftmost leaf. 

Algorithm~\ref{a2} (in Appendix~\ref{pseudo}) shows the details.
All its operations are implemented in constant time in $BP$, and thus \texttt{nextcontained($v$)} is implemented in $\mathcal{O}(1)$ time. The function \texttt{buildL} stays the same, and its cost is dominated by the (at most) $k-2$ calls to \texttt{nextcontained}.  
\end{proof}

Once the set $L$ with the contained nodes of $v$ is built, we can compute the valid forward overlaps of $v$ by following the edges of every $l \in L$ until finding a p-node. We then define:

\begin{itemize} 
  \item \texttt{foverlaps($v$)}: the set of p-nodes whose prefixes overlap a suffix of $v$ of length $\ge m$.
\end{itemize}

Computing the forward overlaps of $v$ by following the edges of every $l \in L$ can be exponential. We devise a more efficient approach that uses $T'$ and the reverse complements of the node labels. We need to define first the idea of bi-directionality in \textit{rBOSS}.

\subparagraph{Simulating bi-directionality. \rm  When building \textit{rBOSS} on $\mathcal{R}^*$, the reverse complement $R_i^{rc}$ of every read $R_i \in \mathcal{R}$ is also included, because there are several combinations in which two reads, $R_i$ and $R_j$, can have a valid suffix-prefix overlap: $(R_i,R_j)$, $(R_i, R_j^{rc})$, $(R_i^{rc},R_j)$, or $(R_i^{rc},R_j^{rc})$, and all must be encoded in $T'$.
An interesting consequence of including the reverse complements is that the topology of the dBG becomes symmetric.}

\begin{lemma}\label{lrc}
The incoming symbols of a node $v$ are the DNA complements of the outgoing symbols of the node $v^{rc}$ that represents the reverse complement of $v$. Further, the outgoing nodes of $v^{rc}$ are the same as the DNA complements of the incoming nodes of $v$.
\end{lemma}

\begin{proof}
Consider the $(k-1)$-length substring $bXc$ of $\mathcal{R}$, and a symbol $a$ that appears at the left of some occurrences of $bXc$. For building \textit{rBOSS}, both substrings $abXc$ and its reverse complement $(abXc)^{rc}=c^{c}X^{rc}b^{c}a^{c}$ are considered. As a result, the dBG node $v$ labeled $bXc$ will have and incoming symbol $a$, and the dBG node $v^{rc}$ labeled $(bXc)^{rc} = c^c X^{rc} b^c$ will have an outgoing symbol $a^c$. Thus, the label of node \texttt{forward($v^{rc}$, $a^c$)} will be $X^{rc}b^{c}a^{c}$, which is the reverse complement of string $abX$, the label of node \texttt{backward($v$,$a$)}. 
\end{proof}

As a result of including the reverse complements of the reads, the cost of computing the incoming symbols of node $v$ becomes proportional to the cost of computing the position of $v^{rc}$ in the \textit{BOSS} matrix.

\begin{theorem}\label{theo:rev_comp}
Computing the position in $M$ of $v^{rc}$ takes $\mathcal{O}(k\log\sigma)$ time. By augmenting \textit{rBOSS} with $s\log s$ extra bits, $s$ being the number of solid nodes, the time decreases to $\mathcal{O}(1)$.
\end{theorem}

\begin{proof}
First, extract the label $lab$ of $v$, then compute its reverse complement $lab^{rc}$, and finally, perform \texttt{backwardsearch}$(lab^{rc})$. The label of $v$ is extracted in time $\mathcal{O}(k\log\sigma)$ with the FM-index, and computing its reverse complement takes $\mathcal{O}(k)$ time. The function \texttt{backwardsearch}, also defined on the FM-index, returns the range of $(k-1)$-length strings in $M$ suffixed by $lab^{rc}$, and it also takes $O(k\log \sigma)$ time. Therefore, computing the position of $v^{rc}$ in $M$ takes $\mathcal{O}(k\log\sigma)$ time. Alternatively, we can store an explicit permutation on the $s$ solid nodes, so that using $s\log s$ bits we find the position of $v^{rc}$ in $M$ in constant time.
\end{proof}

Theorem~\ref{theo:rev_comp} allows us to compute the forward overlaps of $v$ in time proportional to the size of the label of $v$.

\begin{theorem}\label{theo:fovp}
The function {\tt foverlaps} can be computed in $\mathcal{O}(k\log\sigma)$ time.
\end{theorem}

\begin{proof}
First, create $L_{v}=\texttt{buildL}(v)$, and then obtain the reverse complement of the linker node $l=L_{v}[|L_{v}|]$, that is, the one representing the smallest suffix of $v$. Second, compute $l^{rc}$ and search for the range $[i..j]$ in $M$ of the $(k-1)$-length strings suffixed by $l^{rc}$. From the edge symbols in $[i..j]$ follow the dBG path $p_{v^{rc}}$ that spells the label of $v^{rc}$. Finally, every time a solid node $v'$ is reached during the traversal of $p_{v^{rc}}$, report its reverse complement as a forward overlap for $v$. Computing $L$ takes $\mathcal{O}(k)$ time. Both searching for $[i..i]$ and traversing $p_{v^{rc}}$ take $\mathcal{O}(k\log \sigma)$ time. All the shifts between reverse complements take $\mathcal{O}(1)$ time if we use permutations.
\end{proof}

Figure \ref{fig:rc} exemplifies the overlap function. Note that backward overlaps can be obtained by computing \texttt{foverlaps} for the reverse complement of $v$.
The complexity of \texttt{foverlaps} is the same obtained by Simpson and Durbin \cite{simpson2012efficient}.

By using $T'$ and Lemma~\ref{lrc} we can access the topology of the overlaps, and to retrieve extra information from the data that irreductible overlap graphs or dBGs do not have. We formalize this idea as \emph{weighted irreductible overlaps}.

\begin{figure}[t]
\centering
\includegraphics[width=0.75\linewidth]{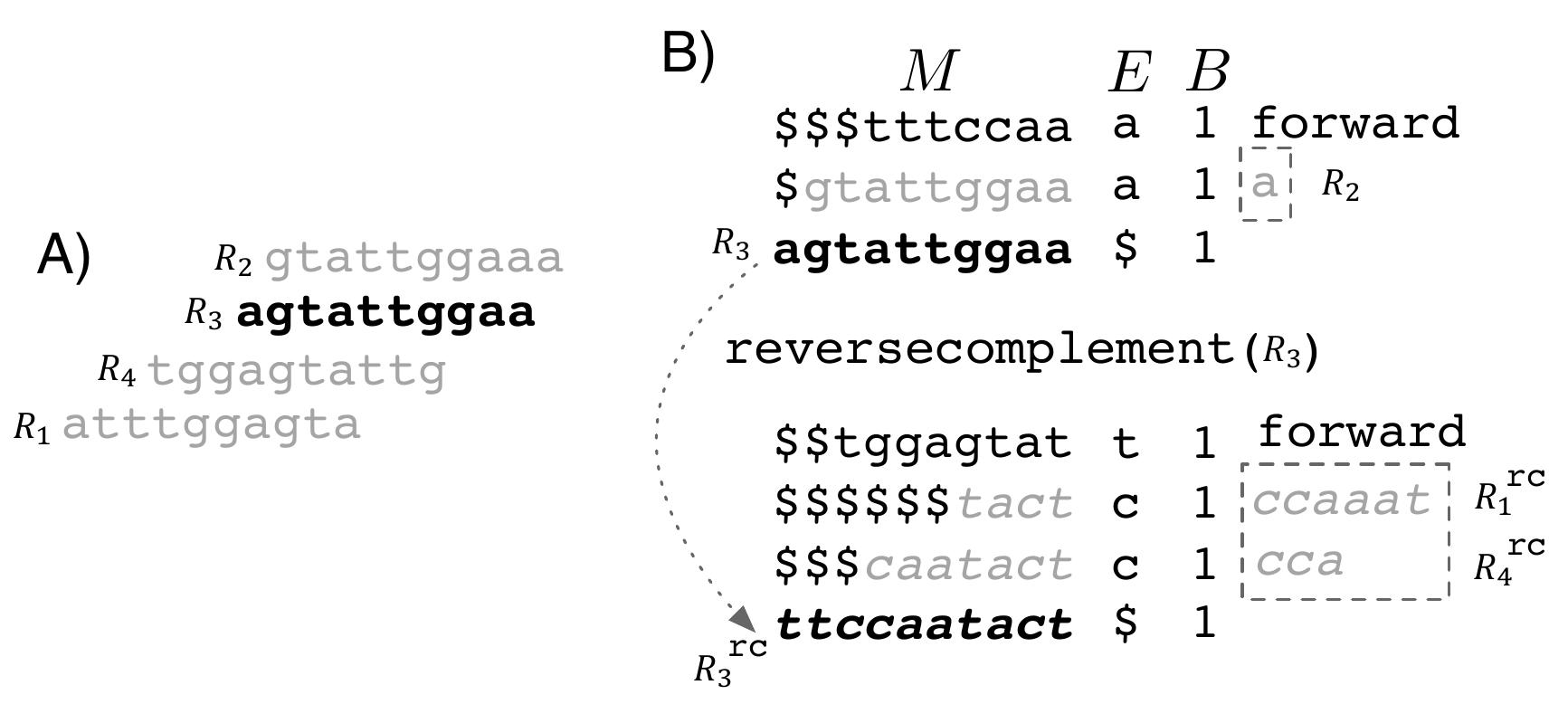}
\caption{Example of the computation of the forward and backward overlaps for sequence $R_3=\mathcal{R}[3]$ of the example of Figure~\ref{fig:rboss}. A) Sequence $R_2$ is a forward overlap for $R_3$ and sequences $R_4$ and $R_1$ are backward overlaps. B) The upper matrix represents the range in $M$ that includes the solid node that represents $R_3$ (row in bold), and its contained nodes (grayed row) and the lower matrix is the range of $M$ that includes the solid node of the reverse complement of $R_3$, $R_3^{rc}$ (row in bold and italic), and its contained nodes (gray rows in italic). Gray symbols to the right of every matrix are the outgoing symbols retrieved from applying \texttt{foward} from every contained node until reaching the next solid node. For the case of $R_3^{rc}$, these solid nodes are $R_{4}^{rc}$ and $R_{1}^{rc}$, the reverse complements of $R_4$ and $R_1$, respectively.} 
\label{fig:rc}
\end{figure}

\subparagraph{Weighting irreductible overlaps. \rm Given an irreductible overlap $v\oplus^{o}v$ between solid nodes $v$ and $v'$, we can use the number of unique transitive overlaps between them as a measure of confidence, \texttt{weight}$(v\oplus^{o}v')<o-m$, for \texttt{label}($v\oplus^{o}v'$). In Figure \ref{fig:rboss_weights} we show different examples in which \texttt{weight}$(v\oplus^{o}v')$ can be helpful to detect patterns in the data. The function $\texttt{foverlaps}(v)$ in our scheme can be modified to return the list of irreductible overlaps for $v$, with their weights included. The idea is as follows: once the range $[i..j]$ is obtained, we form an array $Y$ with the dBG nodes in $[i..j]$ that are not contained by any other node within the same range. The set $Y$ will represent the possible irreductible overlaps of $v$. $Y$ is built in one scan over $[i..j]$ by checking which $T'$ leaves are not the leftmost children of their parent. The weight of every $y \in Y$ is computed as its depth minus the depth of its closest ancestor in $T'$ with more than two children. We do the subtraction because only unique transitive connections count as weights. Every $y \in Y$ and its weighting nodes form a subrange $\range{y}$ in $[i,j]$. We perform a right traversal starting from the outgoing edges of $\range{y}$ to retrieve $p_{v^{rc}}$ as before. In the process, however, one or more elements of $Y$ can be discarded or their weights decreased if they do have a branch spelling the reverse complement of some $l \in L_{v}$. Figure~\ref{fig:rboss_woverlaps} exemplifies the process.}\label{wio}

\begin{figure}[hbt!]
\centering
\includegraphics[width=\linewidth]{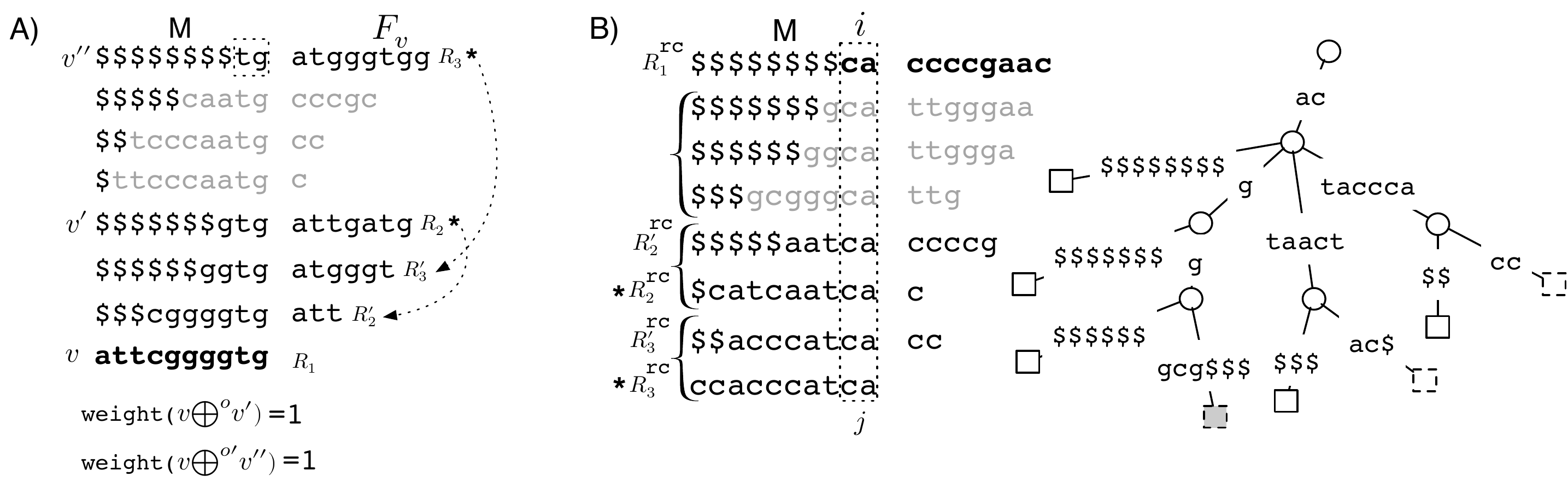}
\caption{Computation of weighted forward overlaps as described in Section \ref{rboss_sec}. A) Sequence $R_1$, represented by dBG node $v$ (bold row in $M$), and its irreductible overlaps, $R_2$ and $R_3$ (asterisks), represented by dBG nodes $v'$ and $v''$ respectively. The transitive overlaps that weight every irreductible overlap are shown with arrows to the right of $M$. The dashed box in the upper left corner of $M$ is $lab=\texttt{llabel}(L_{v}[|L_{v}|])$. The left side of B) is the range $[i..j]$ in $M$ resulting from searching the range of $(k-1)$-length strings suffixed by $lab^{rc}$. The right side of B) is the subtree in $T'$ induced by the dBG nodes in $[i..j]$. Dashed leaves in the subtree are those that are (probably) irreductible overlaps of $v$ (elements in $Y$). The gray dashed leaf of the subtree corresponds to an element $y'$ that was originally added to $Y$ but then discarded because none of the branches starting in its outgoing edges spell the reverse complement of some $l\in L_{v}$. Gray leaves are the ones that contribute to the weight of $y'$. The element $y'$ and its weighting nodes are also represented in the range $[i..j]$ as gray rows. Subranges of $[i..j]$ whose outgoing branches lead to weighted irreductible overlaps of $v$ are shown in curly brackets.} 
\label{fig:rboss_woverlaps}
\end{figure}

\section{Experiments}

We implemented \textit{rBOSS} as a \texttt{C++} library, using the \texttt{SDSL} library~\cite{gbmp2014sea} as a base. In Section~\ref{prems} we stated that vector $E$ can be represented using a Huffman-shaped Wavelet Tree, but our implementation uses run-length encoding~\cite{MN05} to exploit repetitions in the reads. We also include an extra bitmap $S[1..n]$ that marks the position of every solid node in $M$, which speeds up iterating over the solid nodes. We did not include the permutation to compute the reverse complements of the dBG nodes in constant time. Instead, we use \texttt{backwardsearch} as stated in Theorem \ref{theo:rev_comp}. Additionally, we implemented the \textit{VO-BOSS} data structure by modifying our \textit{rBOSS} implementation and merging it with segments of the code\footnote{https://github.com/cosmo-team/cosmo} from Boucher et al.~\cite{BBGPS15}. Our complete code is available at {\tt https://bitbucket.org/DiegoDiazDominguez/eboss-dt/src/master/}. The compilation flags we used were \texttt{-msse4.2 -O3 -funroll-loops -fomit-frame-pointer -ffast-math}.

We used \texttt{wgsim}~\cite{Li2012} to simulate a sequencing dataset (in \texttt{FASTQ} format) from the E.coli genome with 15x coverage. A total of 549,845 reads were generated, each 150 bases long, yielding a dataset of 185 MB\@.   
The input parameters for building \textit{rBOSS} are $k$ and $m$. We used a minimum value of 50 for $k$, and increased it up to 110 in intervals of 5. For every $k$, we used 6 values of $m$, from 15 to 40, also in intervals of 5. This makes up 72 indexes. We also built equivalent \textit{VO-BOSS} instances using the same values for $k$.

\subparagraph{Space and construction time. \rm
The sizes of the resulting \textit{rBOSS} indexes are shown in Figure~\ref{fig:index_sizes}.A, which grow fairly linearly with $k$, at $0.29 + 0.036k$ bits per input symbol (i.e., 50--100 MB for our dataset), and do not depend much on $m$. Figure~\ref{fig:rboss_stats} shows elapsed times and memory peaks during construction. These are also linear in $k$; for example with $m=20$ (the most demanding value) the {\em rBOSS} index for our dataset is built in 4--6 minutes with a memory peak around 2.5 GB. Figure~\ref{fig:index_sizes}.B compares the sizes of \textit{VO-BOSS} and \textit{rBOSS}, showing that \emph{rBOSS} is more than $20\%$ smaller on average.}

\begin{figure}[t]
\centering
\includegraphics[width=0.8\linewidth]{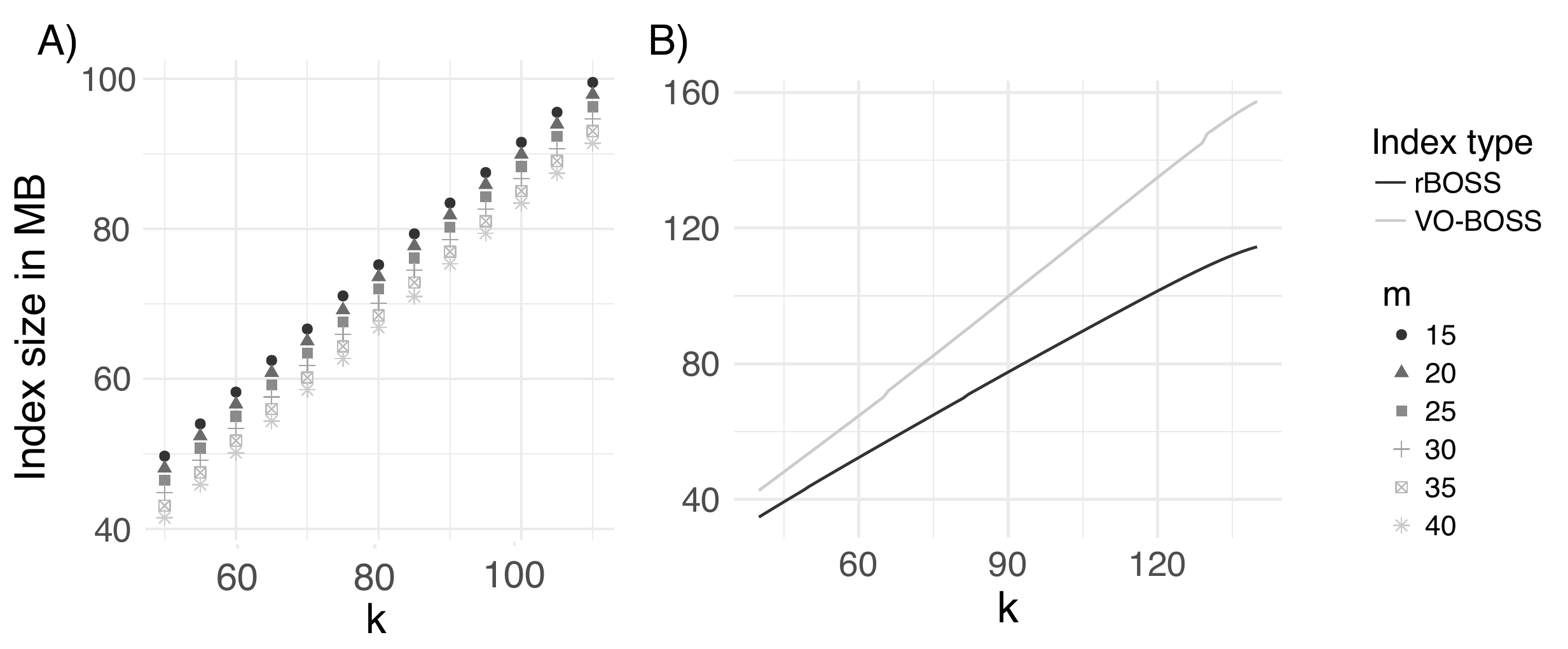}
\caption{Index size statistics. The x-axis gives the values of $k$ and the y-axis is the size of the index in MB. A) Sizes of the \textit{rBOSS} indexes.  Shapes denote the different values of $m$. B) Index size comparison between \textit{rBOSS} and \textit{VO-BOSS}, building the \textit{rBOSS} indexes with $m=30$.} 
\label{fig:index_sizes}
\end{figure}

The space breakdown of our index is given in Figure~\ref{fig:rboss_frac}, and further statistics in Table~\ref{tab:rboss_stats}. The most expensive data structure in terms of space (50\%--65\%) is the BP representation of $T'$. The sequence $E$ uses 20\%--35\%, and the rest are the bitmaps $B$ and $S$.

\begin{figure}[t]
\centering
\includegraphics[width=0.8\linewidth]{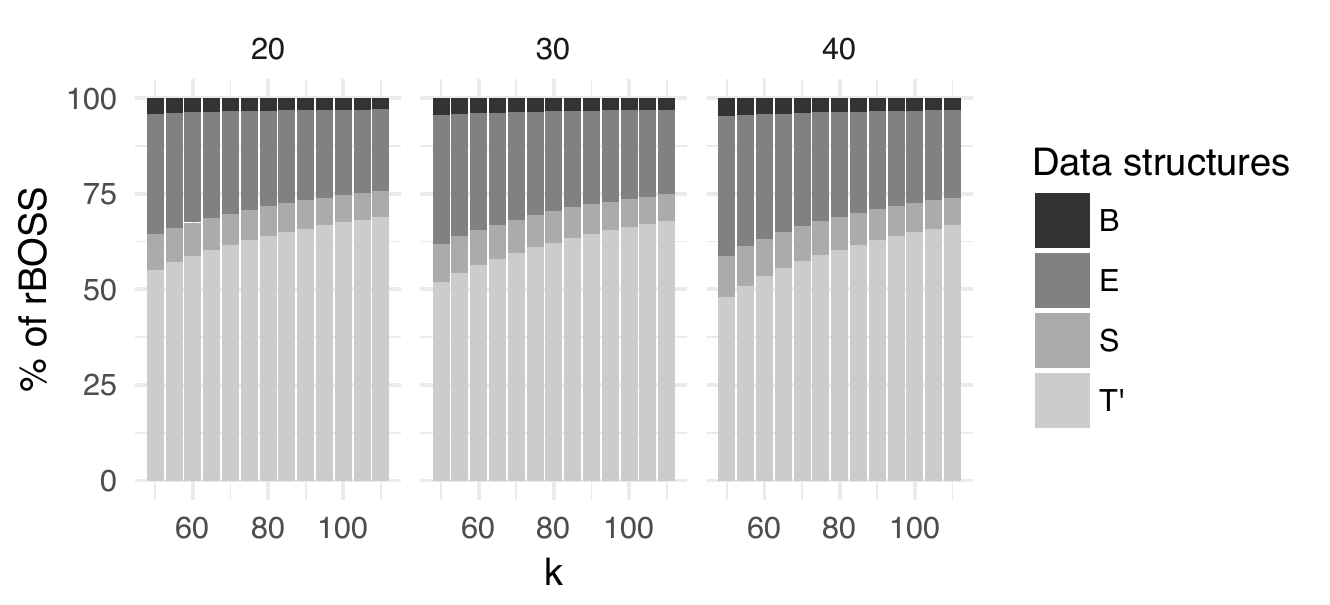}
\caption{Stacked barplot with the percentage that each substructure uses in \textit{rBOSS}. The x-axis shows the value of $k$ and the numbers on top of the plot are the $m$ values.} 
\label{fig:rboss_frac}
\end{figure}

\subparagraph{Time for the primitives. \rm
For every index, we took 1000 solid nodes at random and computed the mean elapsed time for functions \texttt{nextcontained}, \texttt{buildL}, \texttt{foverlaps}. For the \textit{rBOSS} indexes, we also measured the mean elapsed time for \texttt{reversecomplement}. Table~\ref{tab:rboss_funcs} shows the results. Within the \textit{rBOSS} implementation, \texttt{nextcontained} is the fastest operation, with a stable time around 1.5 $\mu$sec across different values of $m$ and $k$. Operation \texttt{buildL} becomes slower as we increase $k$, but faster as we increase $m$. This is expected because the larger $k$, the longer the traversal through $T'$, but if $m$ grows the traversal shortens as well. In all cases, \texttt{buildL} takes under 10 $\mu$sec. The cost of operation \texttt{foverlaps} grows linearly with $k$, but also decreases as we increase $m$, reaching the millisecond. This is much slower than previous operations, dominated by the time to find the reverse complement of the shortest linker node with backward search. Finally, the time of \texttt{reversecomplement} is also a few milliseconds, growing steadily with $k$ regardless of $m$. \\}

Table~\ref{tab:rboss_funcs} also compares \textit{rBOSS} with the  \textit{VO-BOSS} implementation.
All the functions are clearly slower in \textit{VO-BOSS}, by two orders of magnitude for \texttt{next-contained} and \texttt{buildL}, and by a factor around 2 for \texttt{foverlaps}.

%

\begin{table}[t]
\centering
\resizebox{\textwidth}{!}{
\begin{tabular}{cc|cccc|ccc}
\toprule
&&\multicolumn{4}{c|}{\textit{rBOSS}}&\multicolumn{3}{c}{\textit{VO-BOSS}}\\
$k$   & $m$ & \texttt{next-} & \texttt{buildL} & \texttt{foverlaps} & \texttt{reverse-} &\texttt{next-}&\texttt{buildL}&\texttt{foverlaps}\\
& & \texttt{contained} &  &  & \texttt{complement}&\texttt{contained}&& \\
\midrule
50 & 20 & 1.49 & 5.09 & 389.42 & 1226.53 &  225.93 & 804.81 & 825.11\\
50 & 30 & 1.53 & 4.22 & 352.41 & 1209.02 &  216.47 & 581.31 & 802.23\\
50 & 40 & 2.00 & 3.38 & 255.02 & 1226.56 &  191.62 & 337.95 & 770.70\\
70 & 20 & 1.55 & 6.46 & 601.94 & 1620.22 &  311.46 & 1614.49 & 1155.22\\
70 & 30 & 1.57 & 5.82 & 546.53 & 1620.78 &  310.74 & 1382.25 & 1115.33\\
70 & 40 & 1.54 & 5.26 & 517.43 & 1621.98 &  297.23 & 1083.36 & 1080.17\\
90 & 20 & 1.73 & 8.11 & 828.12 & 2013.00 &  374.09 & 2441.96 & 1495.37\\
90 & 30 & 1.58 & 7.35 & 768.83 & 2012.36 &  368.71 & 2211.05 & 1444.93\\
90 & 40 & 1.56 & 6.67 & 714.42 & 2016.41 &  372.76 & 1871.19 & 1398.07\\
110 & 20 & 1.67 & 9.25 & 1088.41 & 2411.10 &  429.86 & 3491.07 & 1865.60\\
110 & 30 & 1.77 & 8.64 & 1014.32 & 2410.03 &  428.17 & 3226.45 & 1801.85\\
110 & 40 & 1.64 & 8.10 & ~942.17& 2414.11 &  436.15 & 2965.48 & 1745.31\\
\bottomrule
\end{tabular}}
\caption{Mean elapsed time, in $\mu$seconds, for the functions proposed in this article for both \textit{rBOSS} and \textit{VO-BOSS}.} 
\label{tab:rboss_funcs}
\end{table} 
  
\subparagraph{Genome assembly. \rm
We implemented a genome assembler on top of \textit{rBOSS} to test the usefulness of the data structure. The algorithm is described in Appendix~\ref{assembly}. We used the same E. coli dataset as before, with a minimum value for $k$ of $60$, increasing
it up to $100$ in intervals of $5$ for building the indexes. For each $k$, we selected $5$ values for $m$, from $30$ to $50$, also in intervals of $5$. The results are shown in Figure \ref{fig:assm_stats}; time and space are again linear in $k$. Using $m=30$, our assembler generates contigs in 
7--14 minutes and has a memory peak of 
70--105 MB, just 
18--21 MB on top of the index itself.   \\ }

Figure~\ref{fig:assm_stats}.C compares the quality of the assembly using variable $k$ and {\em rBOSS}, with $m=30$, versus the corresponding assembly generated with a fixed dBG that uses the same $k$. The dBG indexes were built using the \texttt{bcalm} tool~\cite{chikhi2016compacting}. It is clear that the ability to vary the value of $k$ to compute overlapping sequences as we spell the contigs, also called \emph{maximal paths} (MP) in our algorithm, yields an assembly of much higher quality. Figure~\ref{fig:assm_exp} gives further data on the assembly.

\begin{figure}[t]
\centering
\includegraphics[width=\linewidth]{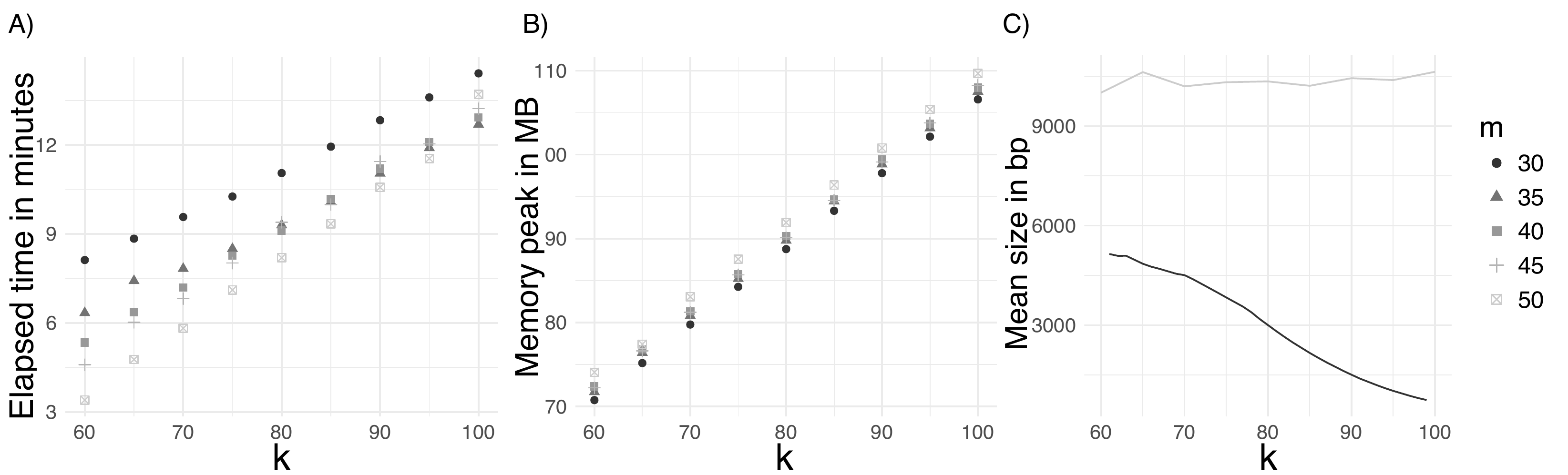}
\caption{Results of genome assembly experiments. A) The elapsed time for the assembly. B) The memory peak achieved during the assembly. C) Comparison of the mean contig size in \textit{rBOSS} (gray line) versus the mean unitig size of a fixed dBG (black line) using the same $k$, both with $m=50$.} 
\label{fig:assm_stats}
\end{figure}

\section{Conclusions and Further Work}

We have introduced {\em rBOSS}, a succinct representation for vo-dBGs (of degree up to $k$) that avoids the $O(\log k)$-bit penalty factor of previous representations thanks to the use of a new structure we call the {\em overlap tree}. This enables the use of $k$ values sufficiently large so as to simulate the full overlap graph, which is an essential tool for genome assembly and other bioinformatic analyses. Our index, for example, can assemble the contigs of 185 MB of 150-base reads, with $k=100$, in less than 15 minutes and within 105 MB\@. 

Our index builds fast, yet using significant space (in our experiment, 6 minutes and 2.5 GB). Future work includes reducing the construction space, even at some increase in construction time. We also aim to reduce the space of $T'$, the most space-demanding component of our index. Preliminary experiments show that the topology of $T'$ is highly repetitive, and that it can be about halved with a grammar-compressed representation~\cite{NO15}. 


The \textit{rBOSS} index can be used for different bioinformatic analyses, not just genome assembly. An example is the detection of single nucleotide polymorphisms. Polymorphisms are usually inferred by first aligning a multiset\footnote{A sequencing data set generated from the DNA of several individuals from the same specie.} of reads to a reference genome and then looking for mismatches between the aligned reads and the genome. This approach is often expensive as it requires much preprocessing. As an alternative, we can build a \emph{colored} version of the \textit{rBOSS} index, that is, we color the reads according to the individual they were generated from, and then search for every read $x$ that meets the following criteria: i) two or more overlaps with heavy weights, ii) two or more colors, and iii) the overlapping reads share one or more colors with $x$, but not among them. Reads meeting these criteria (and their overlapping sequences) are candidates to map polymorphic sites in the genome. We can then align them to the reference genome and check the sequencing quality of their characters to be sure. This idea for inferring SNPs is similar to the one described in \cite{iqbal2012novo}. 

Another possible application is sequencing error correction. In this case, we search for reads whose overlaps have small weights. If for a particular read $x$, all its forward and backward overlaps have very small weights, say $<2$, then it is reasonable to assume that $x$ contains errors, especially if the sequencing qualities of its characters are low.

\bibliography{rboss}
\newpage

\appendix

\setcounter{table}{0}
\setcounter{figure}{0}
\renewcommand{\thetable}{A\arabic{table}}
\renewcommand\thefigure{A.\arabic{figure}}

\section{Figures}

\begin{figure}[ht]
\centering
\includegraphics[width=0.9\linewidth]{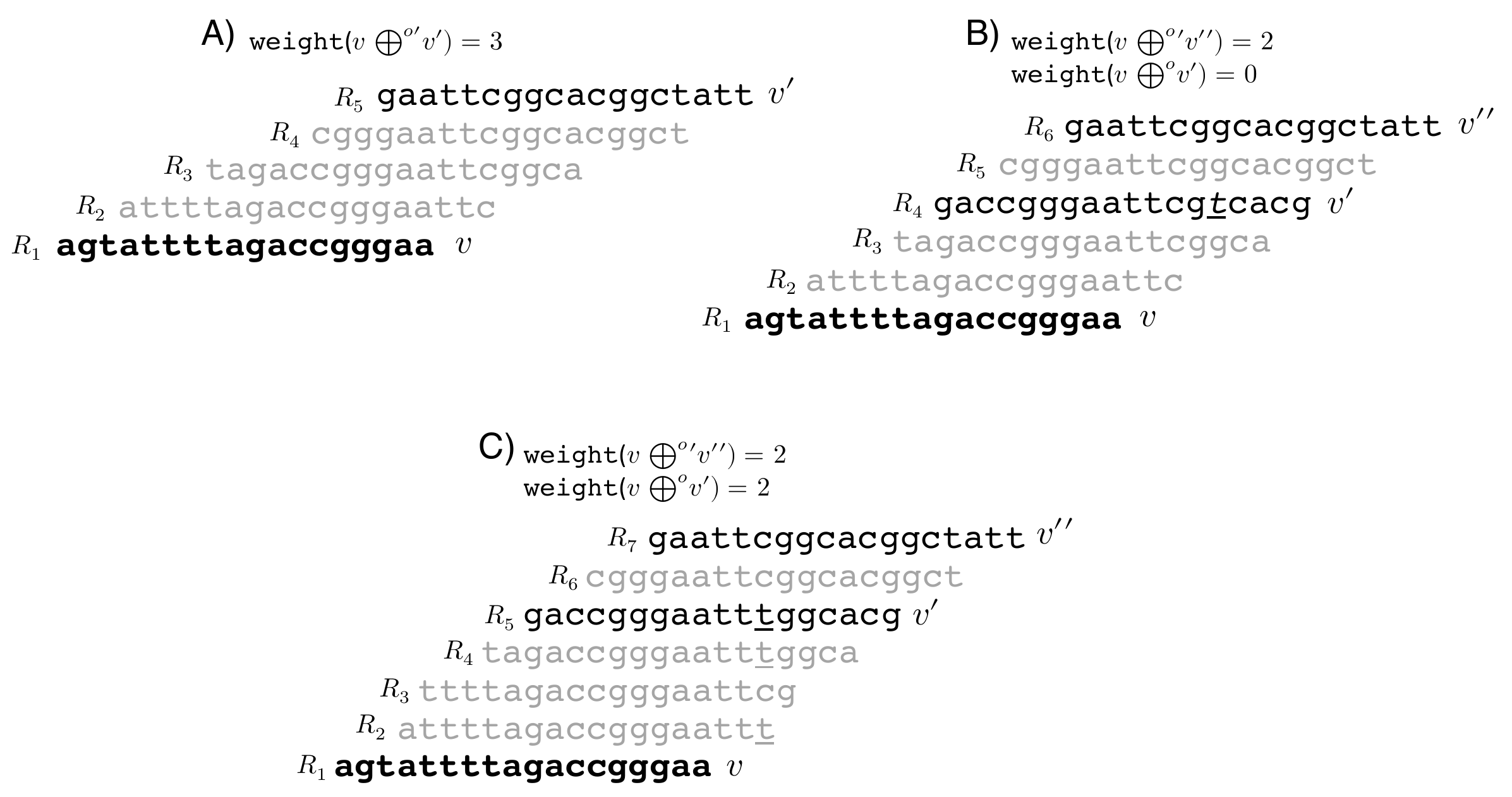}
\caption{Different cases in which the weights of the irreductible overlaps can detect patterns in the data. A) sequence $R_1$ (dBG node $v$) has only one irreductible overlap, with $R_5$ (dBG node $v'$), and there are 3 transitive overlaps between them ($R_2$,$R_3$ and $R_4$ in gray). In this case, there is enough evidence (transitive overlaps) to infer that the string formed by the union of $R_1$ and $R_5$ exists in the input DNA. B) sequence $R_1$ has 2 irreductible overlaps, with $R_4$ and $R_6$ (dBG nodes $v'$ and $v''$ respectively), but the number of unique transitive overlaps between $R_1$ and $R_4$ is zero (\texttt{weight}$(v\oplus^{o}v')=0$), so the most probable option is that $R_4$ contains a sequencing error (italic underlined symbol). C) $R_1$ has two irreductible overlaps, with $R_5$ and $R_7$, and both overlaps have a weight of 2. In this circumstance, it is more probable that $R_1$ belongs to a repeated region or it is next to a genetic variation, if the reads $R_5$ and $R_7$ come from different individuals.} 
\label{fig:rboss_weights}
\end{figure}

\begin{figure}[hbt!]
\centering
\includegraphics[width=\linewidth]{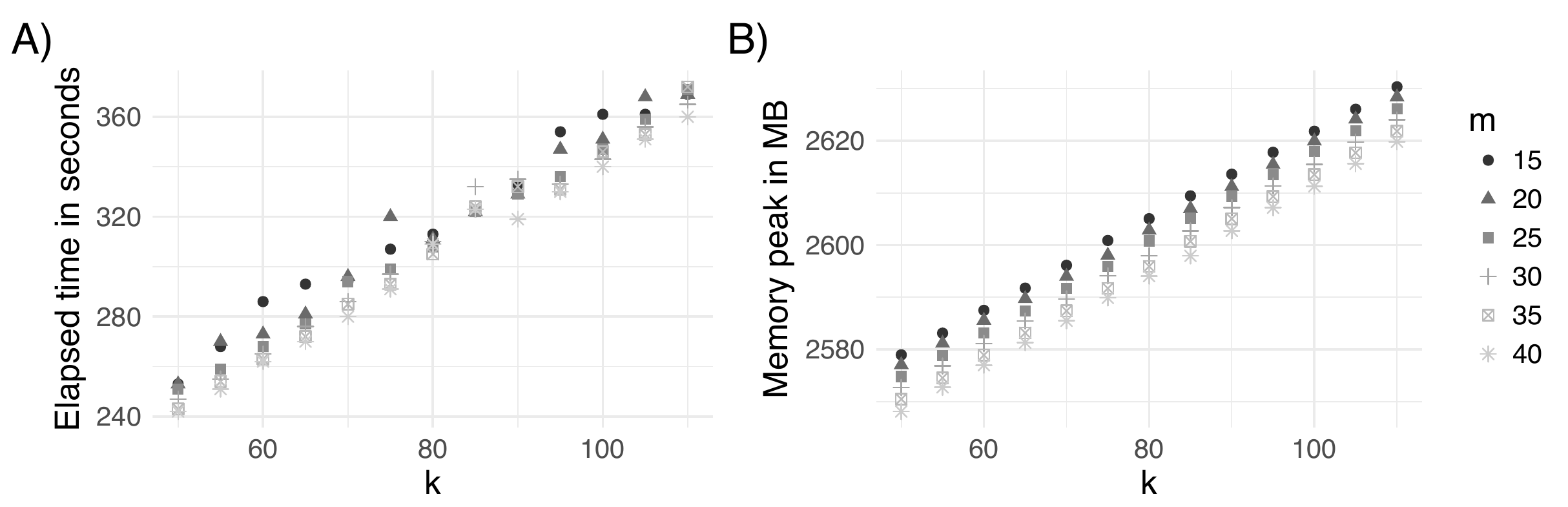}
\caption{Statistics about the construction of \textit{rBOSS}. In all the plots, the x-axis represents the different values used for $k$, and every shape represents a particular value of $m$. The y-axis in A) is the mean elapsed time; in B) it is the memory peak achieved during the construction.} 
\label{fig:rboss_stats}
\end{figure}

\begin{figure}[H]
\centering
\includegraphics[width=0.75\linewidth]{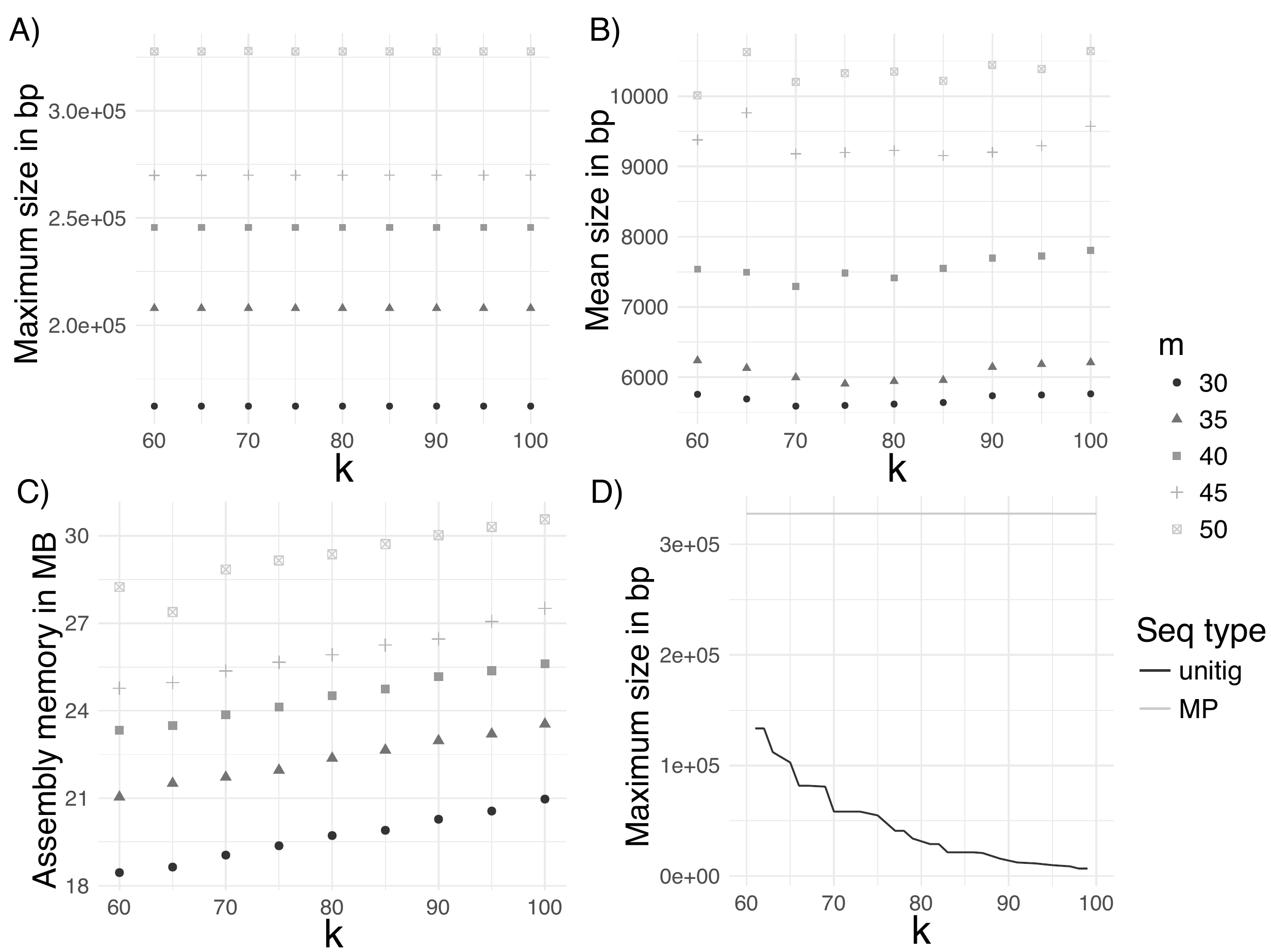}
\caption{Other results of the assembly experiments. In all the figures, the x-axis are the values for $k$. In figures A), B), and C), the shapes are the values of $m$. A) Maximum length achieved for a contig in each \textit{rBOSS} index. B) Mean size for contigs in every \textit{rBOSS} index. C) Difference between the memory peak and the size of the index, that is, the memory used exclusively for assembly. D) Comparison of the longest contig in \textit{rBOSS} (gray line) versus the longest unitig of a fixed dBG (black line) using the same $k$, both with $m=50$.} 
\label{fig:assm_exp}
\end{figure}


\section{Tables}
\begin{table}[H]
\begin{tabular}{cccccccc}
\toprule
\textbf{k}   & \textbf{m}  & \textbf{dBG nodes} & \textbf{solid nodes} & \textbf{linker nodes} & \textbf{edges}  & \textbf{tree nodes} & \textbf{tree int nodes} \\
\midrule
50 & 20 & 48.86 & 9.13 & 39.73 & 50.92 & 78.78 & 29.92 \\
50 & 30 & 48.86 & 9.13 & 39.73 & 50.92 & 68.47 & 19.61 \\
50 & 40 & 48.86 & 9.13 & 39.73 & 50.92 & 58.15 & 9.29 \\
70 & 20 & 69.52 & 9.14 & 60.38 & 71.58 & 120.09 & 50.57 \\
70 & 30 & 69.52 & 9.14 & 60.38 & 71.58 & 109.78 & 40.26 \\
70 & 40 & 69.52 & 9.14 & 60.38 & 71.58 & 99.46 & 29.94 \\
90 & 20 & 90.18 & 9.14 & 81.04 & 92.24 & 161.40 & 71.22 \\
90 & 30 & 90.18 & 9.14 & 81.04 & 92.24 & 151.08 & 60.91 \\
90 & 40 & 90.18 & 9.14 & 81.04 & 92.24 & 140.76 & 50.59 \\
110 & 20 & 110.79 & 9.09 & 101.69 & 112.84 & 202.61 & 91.82 \\
110 & 30 & 110.79 & 9.09 & 101.69 & 112.84 & 192.30 & 81.51 \\
110 & 40 & 110.79 & 9.09 & 101.69 & 112.84 & 181.98 & 71.19 \\ 
\bottomrule
\end{tabular}
\caption{Statistics about the different instances of the dBG graphs generated in the experiments. Except for the first and second column, all the values are expressed in millions. Columns one and two are the values used for $k$ and $m$, respectively, in the \textit{rBOSS} index. Column three contains the total number of dBG nodes at order $k$. Column four show the number of solid nodes and column five the number of linker nodes. Column six is the total number of edges in the dBG (number of symbols in $E$). Column seven is the total number of nodes in the overlap tree and column eight is the number of internal nodes in the overlap tree.} 
\label{tab:rboss_stats}
\end{table}

\section{Pseudocodes} \label{pseudo}

\begin{algorithm}[H]
\caption{Build set $L$ with the linker nodes contained by $v$}\label{alg:voboss_ovp}
\label{a1}
\begin{algorithmic}[1]
\Procedure{{\rm\tt nextcontained}}{$v, m$}
\State$d \gets v.order-1$
\While{$d \geq m$}\label{while_c}
\State$u \gets \mathtt{shorter}(v,d)$\label{shorter}
\If{ $\range{u}\supset\range{v}$}\label{maximal}
\If{$isLinker(v')$ \textbf{and} $|\mathtt{llabel}(v')|=d$}\label{lemma}
\State\textbf{return} $u$ 
\EndIf
\State$v\gets u$
\EndIf
\State$d\gets d-1$
\EndWhile\label{euclidendwhile}
\State\textbf{return} $0$ \Comment{dummy node}
\EndProcedure

\Procedure{{\rm \tt buildL}}{$v,m$}\Comment{$v$ is a vo-dBG node and $m$ is the minimum suffix size}
\State$L \gets \emptyset$
\State$c \gets \mathtt{nextcontained}(v, m)$
\While{$c$>0}\label{while_o}
\State$L \gets L \cup \{\range{c}[1]\}$
\State$c \gets \mathtt{nextcontained}(c, m)$
\EndWhile
\State\textbf{return} $L$
\EndProcedure
\end{algorithmic}
\end{algorithm}

\begin{algorithm}[H]
\caption{Function \texttt{nextcontained} implemented with the topology of $T'$}\label{alg:con2}
\label{a2}
\begin{algorithmic}[1]
\Procedure{{\rm\tt nextcontained}}{$v$}\Comment{$v$ is a vo-dBG node at order $K-1$}
\State$t \gets leafselect(T', v)$\Comment{node in $T'$ mapping $v$}
\State$t' \gets parent(T', t)$
\If{$firstchild(T',t')=t$}\Comment{$t$ is already the leftmost sibling} 
\State$t' \gets parent(T', t')$
\EndIf 
\State$l \gets lchild(T',t')$
\State\textbf{return} $leafrank(T', l)$\Comment{vo-dBG node mapping $l$}
\EndProcedure
\end{algorithmic}
\end{algorithm}

\section{Building \textit{rBOSS}}\label{building}

To build our data structure, we first form the string $R=R_1\$R_1^{rc}..R_r\$R_r^{rc}\#$ over the alphabet $\Sigma \cup \{\$,\#\}$, with size $n'=|R|$, and that represents the concatenation of the reads in $\mathcal{R}^*$. In $R$, symbol $\#$ is the least in lexicographical order. Next, we build the $SA$, $BWT$ and $LCP$ arrays for $\overline{R}$, the reversal of $R$. We use $\overline{R}$ instead of $R$ because the $BWT$ ($E$ in  \textit{BOSS}) contains the symbols to the left of every suffix (node labels in  \textit{BOSS}), but we actually need the symbol to the right when we call \texttt{forward}. After building these arrays, we modify $LCP$ to simulate the padding of the dummy symbols: for every $LCP[i]$, we compute the distance $d$ between $SA[i]$ and the position in $\overline{R}$ of the next occurrence of symbol \$ after $SA[i]$. If $d<k-1$ and $d<LCP[i]$, then we set $LCP[i]=k-1$. To compute $d$ in constant time, we can generate a bitmap $D$, with \emph{rank} and \emph{select} support (see Section~\ref{rs}), that marks in $\overline{R}$ the position of every \$. Thus, $d$ is computed as $select_1(D, rank_{1}(D, SA[i])+1)- SA[i]$. 

The next step is to build $E$. To this end, we traverse $BWT[i]$ for increasing $i$ as long as $LCP[i]\geq k-1$, and in the process we mark in a bitmap $S[1..\sigma]$ the symbols seen so far. When $LCP[i]<k-1$, we append the marked symbols of $S$ to $E$, and also append the same number of bits to $B$, all zeros except the last in each step. We then reset $S$ and restart the traversal of $BWT$.

The final step is to build $T'$. We use an algorithm~\cite{belazzougui2014linear} to build the topology in $BP$ of a tree, modified to discard on the fly the unnecessary nodes. We first compute the virtual suffix tree $ST$ from $LCP$~\cite{abouelhoda2004}, and then traverse it in preorder. For each node $v$, we write an opening parenthesis if it satisfies the restrictions, then we recursively traverse its children, and finally write a closing parenthesis if $v$ satisfied the conditions. 

The $SA$ for $\overline{R}$ can be built in linear time~\cite{karkkainen2006linear, ko2005space, kim2003linear}, and so can $BWT$,~\cite{okanohara2009linear} $LCP$~\cite{kasai2001linear}, and the virtual $ST$~\cite{abouelhoda2004}. Our modifications are obviously linear-time, and therefore the \textit{rBOSS} structure can be built in linear time as well.

\subsection{Rank and select data structures}\label{rs}
\emph{Rank} and \emph{select} dictionaries are fundamental in most succinct data structures. Given a sequence $B[1..n]$ of elements over the alphabet $\Sigma=[1..\sigma]$, $B.\texttt{rank}_{b}(i)$ with $i \in [1..n]$ and $b\in\Sigma$, returns the number of times the element $b$ occurs in $B[1..i]$, while $B.\texttt{select}_b(i)$ returns the position of the $i$th occurrence of $b$ in $B$. For binary alphabets, $B$ can be represented in $n+o(n)$ bits so that \texttt{rank} and \texttt{select} are solved in constant time \cite{Cla96}. When $B$ has $m \ll n$ 1s, a compressed representation using $m\lg\frac{n}{m}+O(m)+o(n)$ bits, still solving the operations in constant time, is of interest \cite{raman2007}. This space is $o(n)$ if $m=o(n)$.

\section{Genome assembly} \label{assembly}

In this section we briefly describe how to use \textit{rBOSS} to assemble a genome. We define some concepts first.

\begin{lemma}
A solid node $v$ is right-extensible (RE) (respectively left-extensible (LE)) if (i) it is a non-s-node with outdegree 1 (respectively, a non-p-node with indegree 1) or (ii) it is an s-node with outdegree $\leq 1$ (respectively p-node with indegree $\leq 1$) and following its outgoing edge (if outdegree is 1) and the outgoing edges of every $l \in L_v$ (respectively its incoming edge, if indegree is 1, and the incoming edges of its backward overlaps) leads to a unique solid node $v'$ in at most $(k-1)-m$ forward operations.
\end{lemma}

\begin{theorem} \label{lem:extensible}
Computing if a solid node $v$ is RE has $\mathcal{O}(k + |L_v|(k-m)\log \sigma)$ worst case time complexity. Computing if $v$ is $LE$ also has $\mathcal{O}(k + |L_v|(k-m)\log \sigma)$ time complexity if we augment rBOSS with $s\log s$ bits, where $s$ is the number of solid nodes. 
\end{theorem}

\begin{proof}
When $v$ is not an s-node, testing if it is RE reduces to checking its outdegree. The other case, when $v$ is a s-node, is harder. First compute $L_v=\texttt{buildL($v$)}$, and then perform a set of operations in batches over the elements of $L_v$, as follows. Regard $L_v$ as a queue. If $L_v[1]$, the linker node that represents the greatest suffix of $v$, has outdegree 0, then remove it from $L_v$. After that, check that each $L_v[i]$, with $i\in[1..|L_{v}|]$, has outdegree 1 and that all the outgoing edges are labeled with the same symbol. If some $L_v[i]$ has outdegree $> 1$ or two or more different symbols are seen in the outgoing edges of $L_{v}$, then return \texttt{false}. If all outgoing edges in $L_v$ have the same symbol $a$, then perform $L_v[i]=\texttt{forward} (L_{v}[i], a)$ for every $i$. Repeat the process until $L_v$ becomes empty, if that happens, then return \texttt{true}. Computing $LE$ is exactly the same process, but first we have to compute $v^{rc}$, the reverse complement of $v$. If we use $s\log s$ extra bits to store the permutation with the reverse complements of the solid nodes, then we can compute $v^{rc}$ in $\mathcal{O}(1)$.
\end{proof}

We also define the concept of right and left maximal paths.

\begin{lemma}
A right-maximal path (respectively, left-maximal) over \textit{rBOSS} is a path where all the nodes are RE (respectively, LE) except the rightmost (respectively, leftmost) node. A maximal-path is the concatenation  of a left-maximal and a right-maximal path.
\end{lemma}

\subsection{Marking non-extensible nodes}

Computing whether a solid node $v$ is extensible  during a graph traversal can be expensive (Lemma~\ref{lem:extensible}), especially if the traversal is exhaustive. The amount of computation can be reduced, however, by computing beforehand which nodes are non-extensible and marking them in a bitmap $N$ of size $s$. Notice that only a small fraction of the nodes will be non-extensible, so $N$ is highly compressible.

There are four cases in which $v$ is non-extensible; (i) it has outdegree $>1$, (ii) there are two or more different outgoing symbols in $L_v$, (iii) the outgoing symbol in $v$ differs from the symbol in $L_v$, or (iv) the computation of the forward overlaps of $v$ yields two or more different irreductible overlaps. To detect non-extensible nodes, we use the topology of $T'$ instead of directly calling the function \texttt{foverlaps}.

We descend on $T'$ in DFS, and every time we reach a leaf $t$ that is the leftmost child of its parent, we append its edge symbols into a sequence $U$ and the leaf rank of $t$ to an array $I$, one append per edge. We also keep track of the different symbols appended into $U$ so far. If after consuming $t$ there are two or more different symbols in $U$, we scan $U$ from right to left until finding the first position $i$ such that $U[i]\neq U[i+1]$. Then, we mark as non-extensible all the solid nodes that contain any prefix of \texttt{llabel($I[i+1]$)} that in turns contains any prefix of \texttt{llabel($I[i]$)} of length $\ge m$. The rationale is that any solid node $v$ containing $I[i+1]$ will also contain $I[i]$ (we now this fact because the DFS order). The problem, however, is that the outgoing symbols of $I[i]$ and $I[i+1]$ differ. Thus, $v$ is a case (ii) non-extensible node. It can still happen that one of the prefixes of $I[i+1]$ is contained by a solid node $v'$, and if it does, then it might happen that $v'$ also contains a prefix of $I[i]$. In such case, $v'$ is a case (iv) non-extensible node, because the elements of $L_{v'}$ will lead to $I[i]$ and $I[i+1]$, which are known to differ in their outgoing edges. To be sure, we must go backwards in $I[i+1]$ marking the solid nodes that contain prefixes of $I[i+1]$, and we stop when $I[i+1]$ does not contain any prefix of $I[i]$ of size $\ge m$.

When a solid node $v$ is reached during the DFS traversal, we first have to check if it was already marked. If it is still unmarked, then we check if it has outdegree more than two ($v$ is a case i), or if it has outdegree 1, but its outgoing symbol differ from the outgoing symbols in $L_v$ ($v$ is case iii).

\subsection{Spelling maximal paths}

Once \textit{rBOSS} and the bitmap $N$ are built, the process of spelling maximal paths can be implemented as an stream algorithm, which is very space-efficient. For every non-extensible node $v$ compute the set $O_v=\texttt{overlaps($v$)}$. We start a forward traversal from each $o_{i} \in O_v$ and continue until reaching a non-extensible node. During the process, append the edge symbols to a vector $F$. After finishing, compute $o_{i}^{rc}$, start a forward traversal from it and continue until reaching the next non-extensible node. As with $o_i$, also append the outgoing edges to a vector $R$. The final string spelled by the maximal path will be $R^{rc} \cdot \texttt{label($o_i$)}\cdot F$. If in either of both traversals, forward or backward, an extensible solid node with outdegree 0 is reached, then call \texttt{nextcontained} and continue through its edges. 

\end{document}